\documentclass[reqno, a4paper]{amsart}

\usepackage{graphicx}%
\usepackage{multirow}%
\usepackage{amsmath,amssymb,amsfonts}%
\usepackage{amsthm}%
\usepackage{mathrsfs}%
\usepackage[title]{appendix}%
\usepackage{xcolor}%
\usepackage{textcomp}%
\usepackage{manyfoot}%
\usepackage{booktabs}%
\usepackage{algorithm}%
\usepackage{algorithmicx}%
\usepackage{algpseudocode}%
\usepackage{listings}%
\usepackage{hyperref}
\usepackage{float}
\usepackage{natbib}
\usepackage[top=1in, bottom=1in, left=1in, right=1in]{geometry}

\usepackage{makecell}%

\newtheorem{theorem}{Theorem}

\newcommand{\reg}[2]{\ensuremath{R_{#1,#2}}}
\newcommand{\eng}[2]{\ensuremath{w_{#1,#2}}}
\newcommand{\diff}[3]{\ensuremath{\delta_#1(#2,#3)}}

\newcommand{\rL}{\diff{0}{k}{0}}
\newcommand{\rR}{\diff{0}{0}{0}}
\newcommand{\rB}{\diff{0}{1}{j}}
\newcommand{\tL}{\diff{1}{k}{0}}
\newcommand{\tR}{\diff{1}{0}{0}}
\newcommand{\tB}{\diff{1}{1}{j}}
\newcommand{\tT}{\diff{1}{1}{0}}

\newcommand{\na}{\ensuremath{\bar{a}}}
\newcommand{\nz}{\ensuremath{\bar{z}}}

\newcommand{\dg}{\ensuremath{\Delta G}}

\newcommand{\edit}[1]{\textcolor{black}{#1}}
\newcommand{\revstart}{\color{black}}
\newcommand{\revend}{\color{black}}

\begin{document}

\title[]{Can geometric combinatorics improve RNA branching predictions?}

\author{Svetlana Poznanovi\'c}\email{spoznan@clemson.edu}
\address{School of Mathematical and Statistical Sciences, Clemson University, Clemson, SC, U.S.A.}

\author{Owen Cardwell}\email{owencardwell@gatech.edu}
\address{School of Mathematics, Georgia Institute of Technology, Atlanta, GA, USA}

\author{Christine Heitsch}\email{heitsch@math.gatech.edu}
\address{School of Mathematics, Georgia Institute of Technology, Atlanta, GA, USA}

\begin{abstract}
Prior results for tRNA and 5S rRNA demonstrated that secondary 
structure prediction accuracy can be significantly improved 
by modifying the parameters in the multibranch loop entropic
penalty function. 
However, for reasons not well understood at the time, 
the scale of improvement possible across both families
was well below the level for each family when considered separately. We resolve this dichotomy here by showing
that each family has a 
characteristic target region geometry, which is distinct from 
the other and significantly different from their own dinucleotide shuffles. 
This required a much more efficient approach to computing the necessary 
information from the branching parameter space, and a new 
theoretical characterization of the region geometries. 
The insights gained point strongly to considering multiple possible 
secondary structures generated by varying the multiloop parameters. 
We provide proof-of-principle results that this significantly 
improves prediction accuracy across all 8 additional families 
in the Archive II benchmarking dataset. 
\end{abstract}

\maketitle


\section{Introduction} \label{sec:intro}

\revstart

The base pairing of an RNA sequence, i.e. the likely secondary structure(s), 
provides important molecular information~\cite{tinoco-bustamante-99} 
and can be critical to generating useful functional 
insights~\cite{miao-etal-20, schneider-etal-23, nithin-etal-24}.
The range of RNA secondary structure prediction approaches is large,
and ever-expanding as the field pursues continued accuracy improvements.
However, the most widely used 
programs~\cite{markham-zuker-08, lorenz-etal-11, reuter-mathews-10},
and their closely related approaches,
are based on optimizing a ``free energy'' score under the
nearest neighbor thermodynamic model (NNTM).

Under the NNTM, a secondary structure decomposes into well-defined 
substructures, 
e.g.\ a stack 
of two consecutive bases pairs or a hairpin loop closed by a single pair.
Each substructure is assigned a score based on the model parameters.
These values are then summed to generate the NNTM score, 
called the ``folding free energy change'' and denoted \dg, 
for the whole structure.
It is worth emphasizing that this \dg\ number is an 
\emph{approximation}
to the actual physiochemical quantity of energy that would be released 
were the full sequence somehow folded in vitro.
Nonetheless, calculating the minimum free energy (MFE) score for a sequence,
and an associated secondary structure, can yield very useful information.
This is particularly true when multiple possible ``suboptimal'' structures 
are considered rather than just a single prediction.

A secondary structure for the given RNA sequence 
is called suboptimal if its \dg\ value 
is close to the MFE score under the current NNTM evaluation.
It has long been recognized 
that prediction accuracy improves when these alternative 
structures are also considered~\cite{mathews-06, wire-review}. 
For instance, the three different
NNTM parameterizations~\cite{jaeger-turner-zuker-89, mathews1999expanded, mathews-etal-04} 
have always reported a ``best suboptimal'' accuracy  
which can be understood as an indicator of the approximation quality.
On average over the set of sequences tested~\cite{mathews1999expanded}, 
the most accurate suboptimal structure found 
(among up to 750 alternative predictions) 
had a \dg\ score within 4.8\% of the MFE.
Hence, NNTM optimization --- particularly when alternative predictions
are also considered --- captures crucial RNA base pairing
information.

It is worth considering how suboptimal structures can be generated.
Stochastic sampling from the Boltzmann 
distribution~\cite{mccaskill-90, ding-lawrence-03}
is the current standard.
However, this approach magnifies small differences in the NNTM score 
since the Boltzmann probability is exponential in the \dg\ values.
Historically, there have been other approaches,
including deterministic sampling~\cite{zuker-89} and
exhaustive generation~\cite{wuchty-etal-99},
which can produce a broader range of alternative predictions.
This can be important when the sampled structures are the
basis for further analysis.
For instance, structural variety is a crucial characteristic
in recent results~\cite{entzian-etal-21} 
motivated by RNA folding kinetics.

We provide here a \emph{mathematical} motivation for generating 
alternative predictions based on a parametric analysis of 
RNA branching~\cite{pmfe_chapt, regions, polystats, bnb}.
Using methods from geometric combinatorics~\cite{pmfe_chapt},
it is possible to identify all optimal predictions under
any parameterization of the entropic branching penalty.
However, due to rapid growth in branching polytope complexity 
with sequence length,
only two types were considered: 
transfer RNA (tRNA) and 5S ribosomal RNA (rRNA).
We demonstrated~\cite{polystats} that re-parameterizing 
the linear penalty for multiloop initiation can significantly 
improve the MFE prediction accuracy on a set of 100 sequences, 
evenly split between the two families.
New parameter combinations were obtained by intersecting 
large cliques in pairwise compatibility for the most accurate
predictions.
Perplexingly, however, the combinations
that were better for tRNA were worse for 5S rRNA, and vice versa.

To determine the exact amount of improvement possible 
on the whole set, as well as each family, 
a branch-and-bound algorithm was developed~\cite{bnb}.
The new results confirm the `exclusive disjunction' conundrum; 
re-parameterizing over the whole set yields an average 
prediction accuracy of 0.66 versus 0.54 for the 
Turner 2004~\cite{mathews-etal-04,turner2010nndb} parameters.
However, parameterizing tRNA alone improves their accuracy from 0.45 to 0.75,
and from 0.64 to 0.73 for just 5S rRNA.  
Moreover, the conclusion that the best prediction accuracy 
via branching re-parameterization is an `either-or-but-not-both' 
situation for these two families was validated by testing
the new parameters on the large number of tRNA and 5S rRNA sequences in 
the Archive II benchmarking dataset~\cite{sloma2016exact, mathews-19} 
from the Mathews Lab (U~Rochester).
Finally, to further confound the situation,
parameter combinations which seem numerically `far' from each 
other can generate statistically indistinguishable accuracies,
even on the small training set of 100 sequences.
This makes the choice of a single `best' parameter combination
seem rather arbitrary.

Here these incongruities are resolved by careful consideration of 
the parameter region geometry.
Preliminary visual inspection had suggested that, for each family,
the regions corresponding to the target conformations --- 
a Y-shape for 5S rRNA and a cloverleaf for tRNA ---  
had a characteristic shape, size, and location.
Moreover, these two sets seemed to be largely disjoint from 
each other.
If confirmed, this would reconcile the potential for improving 
accuracy within a family with the challenge for achieving 
the same level of improvement between the two families. 

Furthermore, although the shape and approximate location were 
maintained under a dinucleotide permutation of each sequence, 
the region size shrank considerably.
This suggested that the former are somehow intrinsic to the
NNTM optimization while the latter is particular to the biological
sequences.
Specifically, the unusually large region size suggests a 
robustness/stability of the family-specific branching conformation.

The results here confirm these initial insights.
This was achieved through two methodological advances. 
First, we present a reduced polytope algorithm which leverages 
the existing code base to more directly compute the parameter 
slice of interest.
The reduction in geometric complexity, and therefore compute time,
makes rigorous randomized comparisons tractable for the first time.
Second, we introduce a linear transformation of the parameter space
which illuminates the underlying geometry.
In this context, we characterize the general shape, and give 
explicit formulas for the size and midpoint location, for the
two parameter regions of interest.
We then use these two advances to give a thorough analysis of 
the geometric characteristics of the 100 training sequences
in comparison to 
a large number (12,000 per family) of random counterparts.

We confirm that the size of the family-specific target regions is extreme.
In sharp contrast to the random sequences, 
the biological ones almost always have a target branching 
conformation that is optimal under a range of parameters
that is typically 4.47 times larger for 5S rRNA and 26.33 times for tRNA.
This size, and the consistency in their midpoint locations, 
explains how a non-empty 
intersection can be found for almost all 50 sequences when  
each family is considered separately.
However, when considered together,
the pairwise overlaps between the different target regions 
are too variable to resolve into a comparable cummulative intersection.
This explains the impossibility of achieving the same level
of prediction accuracy between families as within them.

Given that this parameter geometry is intrinsic to the
NNTM optimization, we propose working with --- rather than
against --- its either-or-but-not-both consequences.
In particular, we explore the potential for improving prediction
accuracy by considering `suboptimal' structures generated 
by varying the branching parameters.
In other words, we consider a range of possible linear approximations
for the multiloop entropic penalty.
The proof-of-principle results for these alternative predictions
are quite promising.
Not only does such an approach achieve the highest accuracy
thus far on the tRNA and 5S rRNA training family counterparts in 
the Archive II dataset,
but it significantly improves the accuracy for all 8 other families
as well.

\revend


\section{Background} \label{sec:back}

\revstart

\subsection{NNTM parameterization}

The quality of the NNTM approximation is fundamentally dependent on the 
parameters used to score the substructures.
Recall that there have been three major 
versions~\cite{jaeger-turner-zuker-89, mathews1999expanded, mathews-etal-04},
referred to as the
Turner 1989 (T89), Turner 1999 (T99), and Turner 2004 (T04) sets. 
The `modern' T99 and T04 values are available online through the nearest 
neighbor database (NNDB)~\cite{turner2010nndb}.
For comparison, the `historic' T89 ones 
can also still be accessed%
\footnote{%
\edit{Via the Internet Archive Wayback Machine 
as ``Version 2.3 free energy paramters for RNA folding'' 
at \url{http://www.bioinfo.rpi.edu/~zukerm/cgi-bin/efiles.cgi?T=37}.
The T99 set was originally called ``Version 3.0''.}}
online.

Different aspects of the NNTM are recognized as more or less well-determined.
The most reliable are the $\Delta G$ values for the 21 distinct base pair 
stacks.
They are firmly grounded in experimental 
data~\cite{freier-etal-86, xia-etal-98}, and did not change
from T99 to T04. 
At the other extreme, the standard form for the branching entropic 
penalty (discussed in detail below) is typically regarded as 
``ad-hoc'' at best~\cite{ward2017advanced}, and involves
three parameters which have changed substantially with each
new version.

In between these extremes are the large number of parameters 
for all the other components of the NNTM: 
the dangling ends and terminal mismatches as well as the hairpin,
bulge, internal, and external loop scores.
For instance, the `small' internal loops,  i.e.\ the ones containing
two base pairs and having 
$1 \times 1$, $1 \times 2/2 \times 1$, or $2 \times 2$ single-stranded
bases on each side, were introduced in T99.
They have more than $7,000$ parameters among them, accounting for the 
5' -- 3' symmetry.
These parameters were updated for T04 based on new experimental 
data (see citations in~\cite{mathews-etal-04}).
The authors indicate that\ldots
\begin{quote}
``{\ldots}measured values are used when available for $1 \times 1$, $1 \times 2$, and $2 \times 2$ internal loops, but approximations are used for most internal loops. The range of measured free energies differs for different types of internal loops. For example, the range is roughly 2 and 6 kcal/mol for $1 \times 3$ \emph{(sic)} and $2 \times 2$ loops, respectively. Evidently, different types of loops require different approximations.''
\end{quote}

The NNTM approximation for 
multibranch loops (also known as junctions or multiloops) has two components:
an initiation term, which is understood as an entropic penalty,
and a stacking term for the favorable intra-loop interactions.
In principle~\cite{turner2010nndb},
the stacking term is ``the optimal configuration of dangling ends, terminal mismatches, or coaxial stacks, noting that a nucleotide or helix end can participate in only one of these favorable interactions.''
In practice,
 the extent to which these stacking interactions
are part of the NNTM optimization can vary. 
For instance, RNAfold~\cite{lorenz-etal-11} has four different dangling 
end options, and the default allows
``a single nucleotide to contribute with all its possible favorable interactions.''
This approximation, known as the ``d2'' 
option\footnote{\edit{For completeness, 
the \texttt{pmfe} code~\cite{pmfe_chapt} 
was based on \texttt{Gfold}~\cite{gtfold-epub} which implemented the 
common dangling options at that time.
The \texttt{pmfe} default, which has been used for all parametric 
analysis results thus far,
corresponds to RNAfold's ``d1'' option, although d2 is also implemented.}},
facilitates efficient partition function calculations but definitely 
can alter the MFE score, and hence the accompanying structural prediction.

The standard form for the initiation term/entropic penalty is:
\[\dg_{\text{init}} = a + b \cdot [\text{number of unpaired nucleotides}] 
+ c \cdot[\text{number of branching helices}] \]
for branching parameters $(a,b,c)$ whose
values have changed considerably over time:
\[ \text{T89} = (4.6, 0.4, 0.1) \hspace*{10ex}
\text{T99} = (3.4, 0, 0.4) \hspace*{10ex} 
\text{T04} = (9.3, 0, -0.6).\]
We particularly note the substantial increase in $a$ and the negative $c$
value in the T04 parameters, as well as that $b = 0$ in both T99 and T04.

It is critical to appreciate that the T89 and T99 ones were 
\emph{learned}, 
i.e.\ they were derived by maximizing the MFE prediction accuracy 
over a set of known structures.
In contrast, the current T04 values are based on experimental data from
small model systems~\cite{diamond-turner-mathews-01, mathews-turner-02}.
There is, however, an important caveat.
The optical melting data was used to parameterize a different function,
where the term involving $b$ above is replaced by 
$b' \cdot [\text{average asymmetry}]$
with $b' = 0.91 \pm 0.19$.
However~\cite{mathews-etal-04}, 
this term is ``neglected'' 
since the average asymmetry computation 
``cannot be accommodated in a dynamic programming algorithm.''
Effectively, then this more accurate model is reduced to the simple
linear one with $b = 0$.

Recall that the linear penalty function was originally introduced for 
computational efficiency~\cite{jaeger1989improved, zuker-mathews-turner-99}. 
However recent results~\cite{ward2017advanced} have demonstrated 
that it out-performs both the logarithmic penalty~\cite{mathews1999expanded}
based on Jacobson--Stockmayer theory
and also one based on polymer theory~\cite{aalberts-nandagopal-10}.
From this the authors conclude that ``the simplest model is best.''

Interestingly, their results included re-parameterizing the logarithmic
and the polyer models, since some of the original values were 
learned using the T99 parameters for the rest of the NNTM.
Their 40 sequence training set was evenly split between 
tRNA and 5S rRNA.  
They found that the newly optimized parameters for both models\ldots
\begin{quote}
``\ldots{i}mprove performance for tRNA, but not for 5S rRNA. This is unexpected, as an equal number of tRNA and 5S rRNA were used for training. Looking at the performance of the parameters on a per RNA basis in the training set suggests that either the performance on tRNA could be increased, or the performance on 5S rRNA could be increased, but not both. We hypothesize that this means that the models themselves are insufficient to describe the space of multi-loop free energies completely.''
\end{quote}

\revend

\subsection{\edit{Parametric Analysis}}

In recent work, tools from geometric combinatorics were useful in performing a complete parametric analysis and quantifying the robustness of the minimum free energy (MFE) prediction as the branching parameters are changed~\cite{polystats}.  The analysis showed that improvement of the prediction for multiple families is possible but difficult because competing parameter combinations favor different families~\cite{bnb}. Whereas previous results focused on isolated parameter combinations, here we explicitly consider the regions of equivalent parameters which yield the same optimal structure. Each of the regions considered yields a structure with a distinct branching pattern. We show that for each of the families, the target branching region has a characteristic geometry which is different from that of the other.

Previous work~\cite{pmfe_chapt} introduced a formulation of the MFE prediction as a linear program focusing on the branching characteristics.  In particular, the \emph{branching signature} of an RNA secondary structure $S$ was defined to be a quadruple $(x,y,z,w)$ where $x$, $y$, and $z$ are the total number, respectively, of multibranch loops (a.k.a. junctions) as well as of single-stranded bases and of helices in those loops, and $w$ is the residual free energy from all the other structural components. The first three correspond to the cumulative number, size, and complexity of the predicted branch points. 
\begin{equation} \label{eq:energy} \Delta G(S)= ax+by+cz+dw, \end{equation} 
where $a$, $b$, and $c$ are the standard components of the NNTM multiloop initiation and $d$ is a scaling term to complete the linear program. Because of this, using methods from geometric combinatorics --- specifically convex polytopes and their normal fans --- it is possible to give a complete analysis of the effect that the branching parameters have on determining the optimal secondary structure.

The \texttt{pmfe} code introduced in~\cite{pmfe_chapt} finds all
vertices of the \emph{branching polytope} for the given RNA sequence.
Each vertex is a branching signature which minimizes Eq.~\eqref{eq:energy}
for some set of real values $(a,b,c,d)$.
The polytope's normal fan is then a conic subdivision of the 4d parameter space
with maximal cones containing all parameters that yield the same optimal
signature.
In the NNTM model, the last parameter $d$ is fixed at $1$, and hence
a parametric analysis need only analyze the $d=1$ slice of this
normal fan.
In this way, we find the optimal signatures, and their associated
secondary structures, for every combination of possible NNTM
branching parameters.

In~\cite{polystats}, this geometric combinatorics approach gave a
complete analysis of the prediction accuracy, stability, and robustness
for a diverse set of 50 tRNA and 50 5S rRNA sequences.
By analyzing all possible optimal secondary structures for each sequence,
we found that every one had a high accuracy prediction for some
combination of $(a,b,c,1)$ values.
However, it was not possible to achieve this level of improvement
for all sequences with the same parameter triple
--- not even within the same family.

Subsequently, a branch-and-bound algorithm~\cite{bnb} showed
that there exist parameters that significantly improve prediction accuracy
for the tRNA test sequences and other combinations that work for the
5S rRNA ones.
Moreover, the best possible improvement over both families simultaneously
is significantly lower.
These conclusions were tested against a much larger dataset (Archive II),
and found to remain valid.
Hence, it was well-understood that changing the default NNTM
branching parameters
can significantly improve secondary structure prediction accuracy,
but not why different combinations are required for tRNA and for
5S rRNA.

To resolve this question, it was determined to focus on the $b = 0$
slice of the $d = 1$ normal fan hyperplane.
This agrees with the current NNTM branching parameters and is supported
by results~\cite{bnb} showing that different $b$ values
near zero have equivalent maximum accuracy improvements.
Moreover, it allows to focus clearly on the geometry of the optimal
parameter regions in the $(a,c)$ plane, particularly the two which
correspond to the target branching signatures for tRNA and for 5S rRNA.
We find large overlap in the target region parameters within each family,
but that these family-specific values occupy essentially distinct
parts of the parameter space from each other.
Furthermore, we show that this region geometry is a true biological
signal by comparison with a large randomized dataset.

Since computing the branching polytope for the sequence lengths we consider takes several hours, we develop an algorithm for computing the $b=b_0, d=1$  slice for any fixed value $b_0$ and use it to generate our data. \edit{Our algorithm leverages the existing code but computes the desired slice without computing the branching polytope.} We show that the location and geometry of the target regions depend on the differences of the residual energies corresponding to structures to close branching patterns and we give exact formulas for the target branching of tRNA and 5S rRNA sequences. Analysis of the distributions of the differences explains the previously observed large overlap of the same target regions for the homologous sequences. Moreover, results show that the target regions are present with higher frequency and bigger area for the biological sequences compared to random ones. We give details about which differences of residual energies have the biggest effect on the difference in the geometry of the regions. Finally,  based on the insights from the data analysis, we consider secondary structure prediction by generating multiple structures which correspond to branching parameters in the neighborhood of the target regions for tRNA and 5S RNA. We show that if the $(a,c)$ branching parameters are allowed to vary over a relatively small region, we get a measurable increase in the prediction accuracy for \edit{all} additional families in the Archive II benchmarking data set.


\section{Methods}

\subsection{Algorithm for slice computation} \label{subsec:algorithm}

An RNA secondary structure refers to the base-pairing interactions within a single RNA molecule, which gives rise to specific shapes or motifs: helices made of stacked base pairs separated by different types of loops (bulges, internal and hairpin loops, and junctions). The junctions are regions where 3 or more helices converge. Recall that the signature of an RNA secondary structure is the  quadruple $(x,y,z,w)$ where $x$, $y$, and $z$ are the total number, respectively, of junctions, single-stranded bases and helices in those loops, and $w$ is the residual free energy from all the other structural components. For a given RNA sequence $s$, its branching polytope $P_B$ is the convex hull of the set of signatures that correspond to secondary structures over $s$.   For a face $F$ of a polytope $P$, its corresponding normal cone is the set of all vectors $\textbf{A}$ such that any point $\textbf{v} \in F$ minimizes the product $\textbf{A}\textbf{v}$. The normal fan of a polytope $P$ is the collection of its normal cones and is denoted by $\mathcal{N}(P)$. The full dimensional cones of $\mathcal{N}(P)$ correspond to the vertices of $P$. 

Recall that $d$ is a ``dummy'' parameter introduced to complete the linear program, and only biologically relevant when $d = 1$. The $d=1$ slice of $\mathcal{N}(P_B)$ is a subdivision of the $a,b,c$ - parameter space in convex regions, so that branching parameters from the same region yield the same optimal branching signature. For our data analysis, we also specialize $b = 0$, which is consistent with both the prior T99 and current T04 branching parameters.

Prior results had characterized the infinite
regions theoretically~\cite{regions} and both types
computationally~\cite{polystats}.
It was shown that the cones are quite ``thin'' in the $b$ 
dimension~\cite{polystats},
but also that this does not make a statistically significant
difference to the prediction accuracy~\cite{bnb}. 
In other words, although individual predictions are not necessarily
robust to small changes in $b$, there are $(a,c)$ values with 
equivalent accuracy on average whether $b \leq 0$ or $b \geq 0$.

In this section we describe an algorithm for computing the $b = b_0, d=1$ slice (from now on a $b_0$-slice) of $\mathcal{N}(P_B)$, for a fixed value $b_0 \in \mathbb{R}$, which \edit{leverages the existing code for polytope computation but} avoids computing $P_B$ itself. The flow of the algorithm is represented in Figure~\ref{fig:new_flow_figure}. It uses the same main components used in the \texttt{pmfe}~\cite{pmfe_chapt} branching polytope calculations with some modifiers. In particular, the \texttt{iB4e} module we use is an adaptation of the software developed by Huggins~\cite{huggins2006ib4e}. In our case, this module internally incrementally builds a 3d polytope, denoted in the explanation below by $P_3$, by extending it in various directions as in the Beneath-and-Beyond algorithm~\cite{grunbaum1967convex}. The main idea behind the Beneath-and-Beyond algorithm is to build a desired polytope by starting with a small initial approximation and then try to confirm the facets of the approximation by checking that there are no vertices 'beyond' it. When a new vertex is found (and thus the confirmation of the facet fails), it is used to refine the approximation. 

The polytope $P_3$ that our algorithm builds internally is the image of the branching polytope $P_B$ under the transformation
$T: \mathbb{R}^4 \to \mathbb{R}^3$ given by \[T(\textbf{v}) = \begin{bmatrix} 1 & 0 & 0 & 0 \\ 0 & 0 & 1 & 0 \\ 0 & b_0 & 0 & 1 \end{bmatrix} \textbf{v}\]
and is not output to the user. While $P_3$ is not output itself, the algorithm outputs signatures which are preimages of the vertices of $P_3$. We denote the convex hull of those signatures by $P_4$. The steps of the algorithm are illustrated in Figure~\ref{fig:new_flow_figure} and the arrows refer to the following steps:
\begin{itemize} 
\item[(a)] Each time  \texttt{iB4e} tries to confirm a facet of $P_3$, it outputs a vector  $\textbf{p} = (a,c,d)$.
\item [(b)] The  \texttt{find-mfe} module finds an MFE structure for the vector $\textbf{p'} = (a, b_0d, c, d)$.
\item [(c)] The \texttt{scorer} module finds the $(x,y,z,w)$ signature of the MFE structure.
\item [(d)] The vector $(x,z,w+b_0y)$ is passed back to \texttt{iB4e} to either confirm the facet with normal $\textbf{p}$ or to extend $P_3$ by a new vertex.
\item [(e)] If the facet is not confirmed,  $(x,y,z,w)$ is stored in the .rnapoly file.
\end{itemize}

\begin{figure}[htbp]
  \centering
 \includegraphics[width = .75\linewidth]{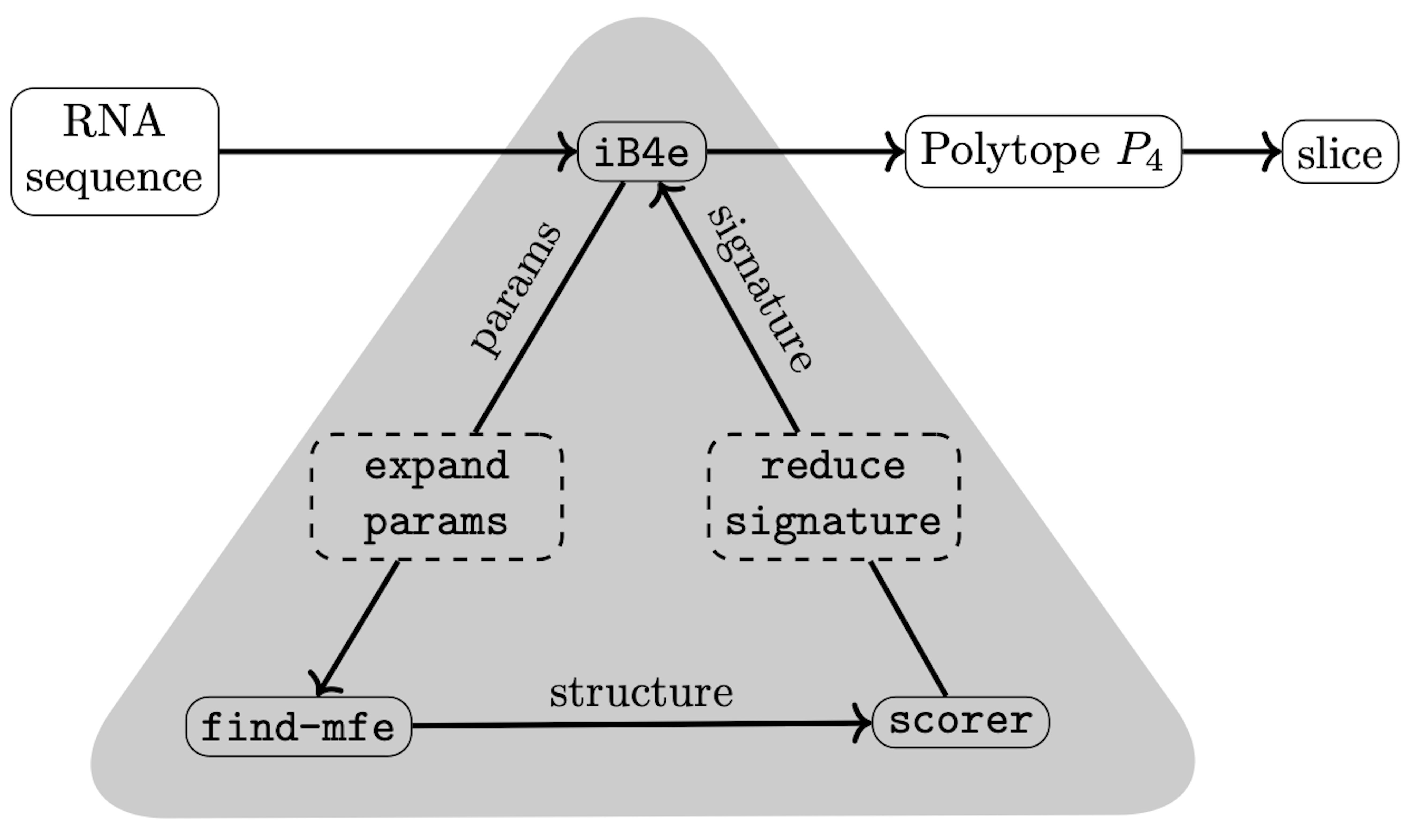}
\caption{Control flow of computing the $b_0$-slice of $\mathcal{N}(P_B)$. The dashed boxes represent modifications introduced to the algorithm for computing the branching polytope $P_B$~\cite{pmfe_chapt}.}
  \label{fig:new_flow_figure}
\end{figure}

The final .rnapoly file contains points whose convex hull is a 4d polytope $P_4$.

\begin{theorem} The polytope $P_4$ has the property
\[ \mathcal{N}(P_B) \cap \{b=b_0, d=1\} = \mathcal{N}(P_4) \cap \{b=b_0, d=1\}. \] 
\end{theorem}
\begin{proof}
Since \[(a, b_0d, c, d) \cdot (x, y, z, w) = ax + cz + d(w+ b_0y),\] the vertices of $P_4$ correspond to MFE minimization by vectors $(a, b_0d, c, d)$ for all triples $(a,c,d)$ by construction. Let $(a',c')$ be in the $b_0$-slice of $\mathcal{N}(P_4)$. Then $(a',b_0,c',1)$ is in $\mathcal{N}(P_4)$ corresponding to a signature $(x, y, z, w)$, which means $(a',c')$ is in the $b_0$-slice of $\mathcal{N}(P_B)$ corresponding to the same signature. Therefore, \[\mathcal{N}(P_4) \cap \{b=b_0, d=1\} \subseteq \mathcal{N}(P_B) \cap \{b=b_0, d=1\}.\]The converse inclusion follows similarly.
\end{proof}

The complexity of the Beneath-and-Beyond algorithm for constructing a polytope is $O(V + F)$, where $V$ is the number of vertices and $F$ is the number of facets of the polytope being computed (since the objective vectors are chosen so that after each iteration, either a new vertex of the convex hull is found or a facet of the convex hull is confirmed the method requires running the LP solver for no more than $O(V + F)$ objective vectors). Compared to $P_B$, the new polytope $P_4$ has a significantly fewer number of vertices and facets. This results in significantly faster computations of the desired slice. Including the step required to construct a polytope object that can be analyzed (for this we used \texttt{SageMath}~\cite{sagemath}), the computation time reduced from $\sim$50 min to $\sim$83 s for tRNA and from $\sim$11 h to $\sim$8 min for 5S rRNA. The code for computing a $b_0$-slice is available at: \url{https://github.com/gtDMMB/pmfe2023}.

\subsection{Slice geometry} \label{subsec:geom}

We consider the normal fan associated with an RNA branching polytope,
and investigate finite regions in the two-dimensional 
slice obtained by specializing the $b$ and $d$ parameters.
As with the previous algorithmic results, the transformation of 
the parameter space is presented for an arbitary $b_0$-slice, 
i.e.\ the $(a,c)$ plane obtained for any fixed $b = b_0$ (and $d = 1$).
We then characterize the regions corresponding to the 
target branching signatures for tRNA and 5S rRNA, focusing on 
the geometry for $b = 0$.

\subsubsection{Excess branching formulation of linear program}

We begin by reformulating the optimization to better
illuminate the slice geometry.
Since every junction must contain at least three helices,
$z = 3x + \nz$ where the new variable $\nz \geq 0$ denotes 
the total `excess' branching beyond the $3x$ minimum.
Then
\[ (a,b_0,c,1) \cdot (x,y,z,w) = (\na,b_0,c,1) \cdot (x,y,\nz,w)\]
where the new parameter $\na = a + 3c$ is the minimum penalty per junction.

\newcommand{\figtab}[1]{\includegraphics[width = 2.5in]{#1}}

\begin{figure}
\centering
\begin{tabular}{cc}
\figtab{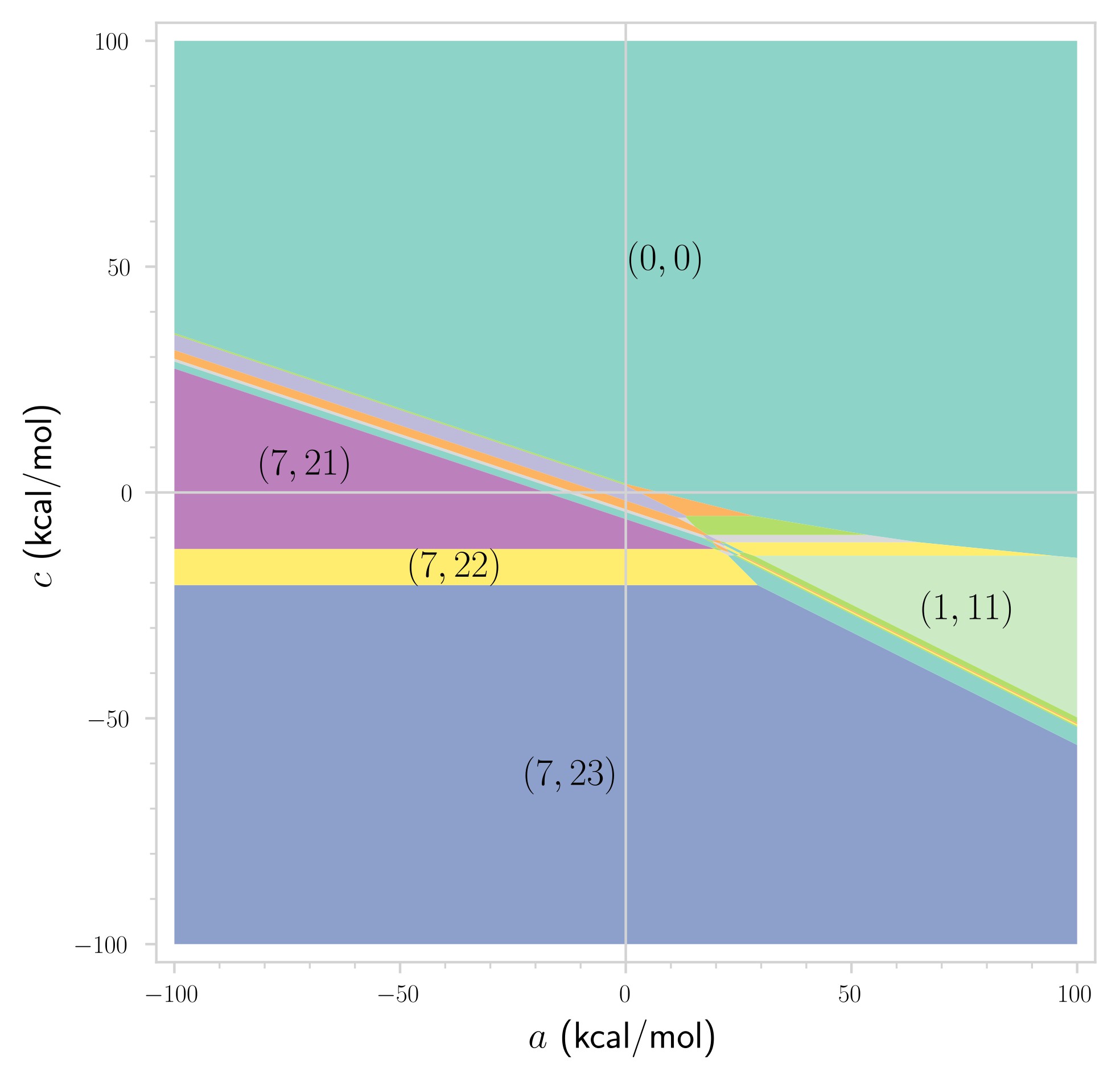} & \figtab{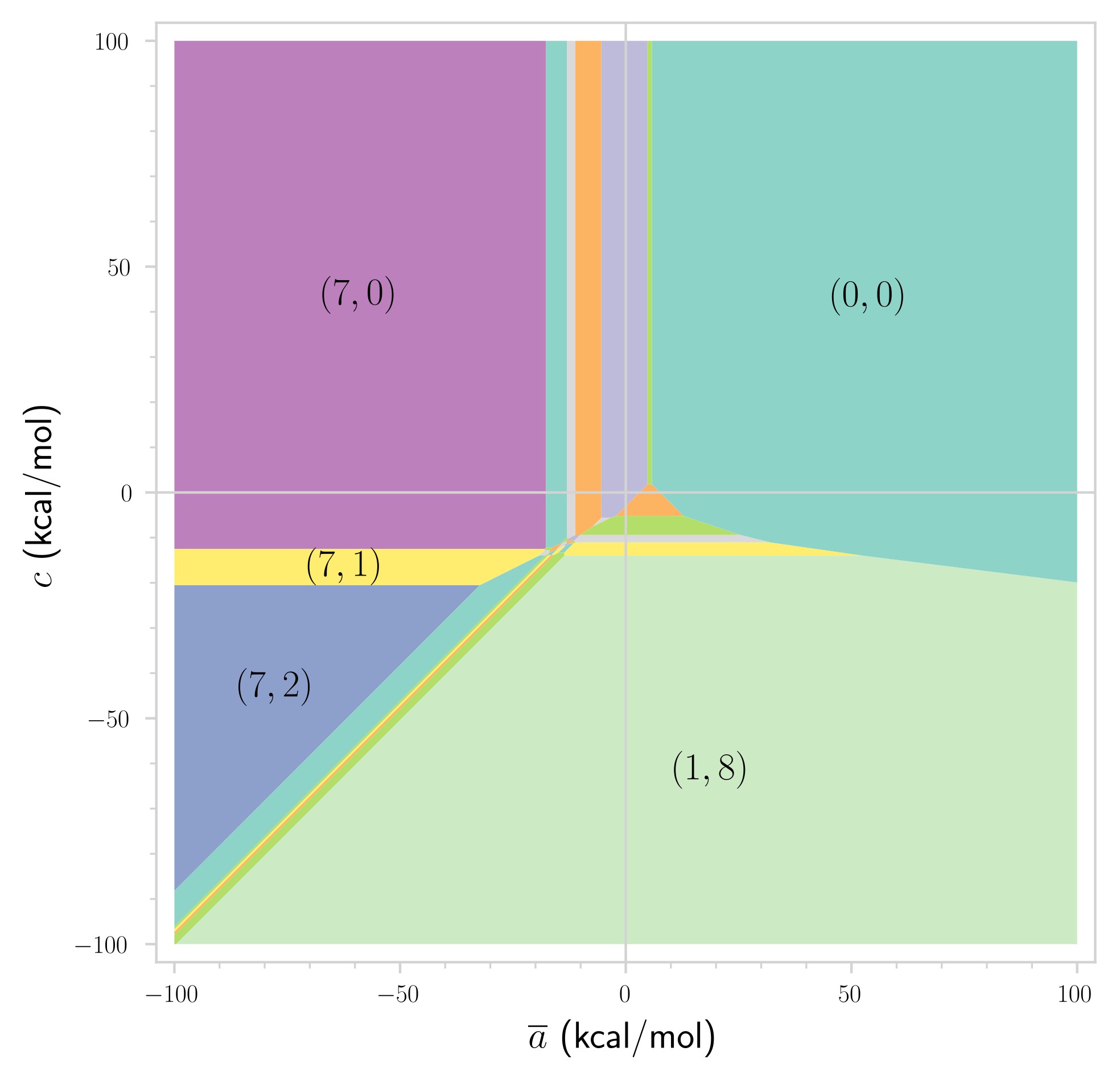}
\end{tabular}
\caption{%
\edit{
The excess branching transformation is illustrated by 
two $b = 0$, $d = 1$ parameter slices for a biological tRNA sequence.
Corresponding regions are colored the same, and larger ones are 
labeled with the appropriate reduced signature.
For example, the $(7,21)$ purple triangle in the original $(a,c)$ plane
on the left is transformed 
into the $(7,0)$ rectangle in the $(\na,c)$ plane on the right.
The parameter range used is only to illustrate the mathematical 
transformation, and certainly not biologically meaningful.
}
\label{fig:unskew}}
\end{figure}

As seen in Figure~\ref{fig:unskew}, this transformation eliminates 
the 1/3 skew in the original $(a,c)$ plane. 
Results from~\cite{regions} 
--- appropriately reinterpreted --- still hold.
As before, all sequences have a unique region $(0,0,0,w_0)$ with 
the minimum number of junctions, the minimum (total or excess) branching,
and the minimum residual free energy $w_0$ over all such structures.
This infinite region always dominates the upper-right quandrant.

There is at most one optimal region with a given number of 
junctions and branches in each $b_0$-slice.
Call these two values the reduced signature for that region, and 
denote the region  as $\reg{x}{z}$ or $\reg{x}{\nz}$ as appropriate.
In the $(a,c)$ plane, adjacent $\reg{x}{z}$ and $\reg{x'}{z'}$  
are separated by a horizonal line if and only if $x = x'$.
This remains true in $(\na, c)$.
Moreover, now the dual holds as well;
adjacent regions $\reg{x}{\nz}$ and $\reg{x'}{\nz'}$  
are separated by a vertical line if and only if $\nz = \nz'$.

There is a maximum number of junctions, denoted $x_{max}$, over the
branching polytope.
We assume $x_{max} > 0$, which holds 
except for extreme cases (such as {\scshape ggggcccc})
\edit{which have no possible structures containing any junctions}.

There is always an infinite region $(1, \nz_{max}(1))$, 
that is a region where the number of excess branches for the single junction
is the largest possible over the entire polytope.
This now dominates the lower $(\na, c)$ half-plane  
pictured in Figure~\ref{fig:unskew}.
It is adjacent to $(0,0)$ since crossing a 
region boundary moving horizontally to the right must decrease $x$.
Likewise, moving vertically upward must decrease $\nz$ if a
boundary is crossed. 
The slope of the boundary is determined by 
what happens to the other variable.
In particular, if regions $(x,y,\nz,w)$ and $(x',y',\nz', w')$
are adjacent, then their boundary line in the $b_0$-slice is 
\[ \Delta_x \na + \Delta_{\nz} c = - \Delta_y b_0 - \Delta_w \]
for $\Delta_x = x - x'$, etc.

Since an excess branch can always be removed from a structure,
there exist reduced signatures (not necessarily optimal)
for $(1,\nz)$ with $0 \leq \nz \leq \nz_{max}(1)$.
When these signatures are optimal, they dominate the finite regions
in the $b_0$-slice and form the left boundary of \reg{0}{0}.

When there exists a signature $(x_{max}, y, 0, w)$, that is 
when $\nz_{min}(x_{max}) = 0$, then 
the other infinite regions in the upper half-plane are
vertical `strips' $(x, 0)$ with $0 < x < x_{max}$.
This is true for all polytopes analyzed here, but can fail 
for sufficiently pathological\footnote{
For instance,
{\scshape 
gggaaacgaaaccggaaacgaaaccggaaacgaaaccc}
has $x_{max} = 4$ but $\nz_{min}(x_{max}) = 1$ while 
for $\nz_{min} = 0$, $x_{max}(\nz_{min}) = 3$.
We also note that $z_{max}(x_{max})$ is not always $z_{max}$,
e.g.
{\scshape
gggaaacgaaacgaaacgaaaccggaaacgaaacgaaacgaaaccc}
has $x_{max} = 4$ and $z_{max} = 13$ but $x_{max}(13) = 3$
and $z_{max}(4) = 0$.
}sequences. 

The other infinite regions in the lower-half plane 
are $(x_{max}, \nz_{max}(x_{max}))$,
whose skewed version dominates the $(a,c)$ one,
and a progression of $(x_i, \nz_{max}(x_i))$ 
with $1 =  x_0 < x_{1} < \ldots < x_{max}$.
As with the vertical strips in the upper half-plane, the $x$ 
values are not necessarily contiguous.
While the slope of their boundary lines is often 1,
smaller values are not uncommon\footnote{
Previously~\cite{regions}, we had noted that 
{\scshape
accccgacccuuuucccagccca}
violates the common progression 
$\nz_{max}(x) \leq \nz_{max}(x-1) - 1$ 
since $\nz_{max}(2) \geq 0$ but $\nz_{max}(1) = 0$.
The progression holds when the $\nz_{max}(x)$ structure is 
``planted,'' i.e. has an external loop with only one branch.}
for biological sequences.
Finally, if there is an $(x_{max}, \nz)$ region with 
$\nz_{min}(x_{max}) < \nz < \nz_{max}(x_{max})$, then 
this will be a horizontal infinite `strip' as 
with the $(7,22)/(7,1)$ one in Figure~\ref{fig:unskew}.

\subsubsection{Characterizing tRNA and 5S rRNA target regions}
\label{subsec:character}

The reduced excess branching signature for the classic tRNA 
cloverleaf secondary structure is $(1,1)$, while the Y-shape
5S rRNA one is $(1,0)$. 
We refer to these as the target region for the respective family.
We will show that their geometry typically
displays a clear biological signal, as pictured in Figure~\ref{fig:target},
in contrast to the randomized comparisons.
For clarity, the results are presented for $b = 0$, but it is 
straightforward to generalize to arbitrary $b_0 \neq 0$.

As seen in Figure~\ref{fig:target}, it is possible to have 
a $b_0$-slice with \reg{1}{0} but not \reg{1}{1} and vice versa.
In general, many of the random sequences have both and some neither.
If a target region exists, then its right boundary is always \reg{0}{0}.
The left boundary is \reg{k}{0} for $k > 1$ and bottom is 
\reg{1}{j} where $j \geq 1$ for \reg{1}{0} and $j \geq 2$ for \reg{1}{1}.
If these are the only boundaries, then \reg{1}{0} is an 
(infinite) rectangle and \reg{1}{1} a triangle.
If \reg{1}{1} is bounded above by \reg{1}{0}, then it's a trapezoid.

We note that the lower left corner of the target region can 
have additional boundary edges. 
For simplicity, however, we do not consider those in our area 
approximation.
We show in Section~\ref{subsec:geocon} that the geometric formulas
below are a very good approximation to the area and location of 
the target regions.
Moreover, this permits us to explicitly consider the different 
residual energy factors,
and determine which is the strongest biological signal.

\begin{figure}[t] 
\centering
\includegraphics[width=0.75\linewidth]{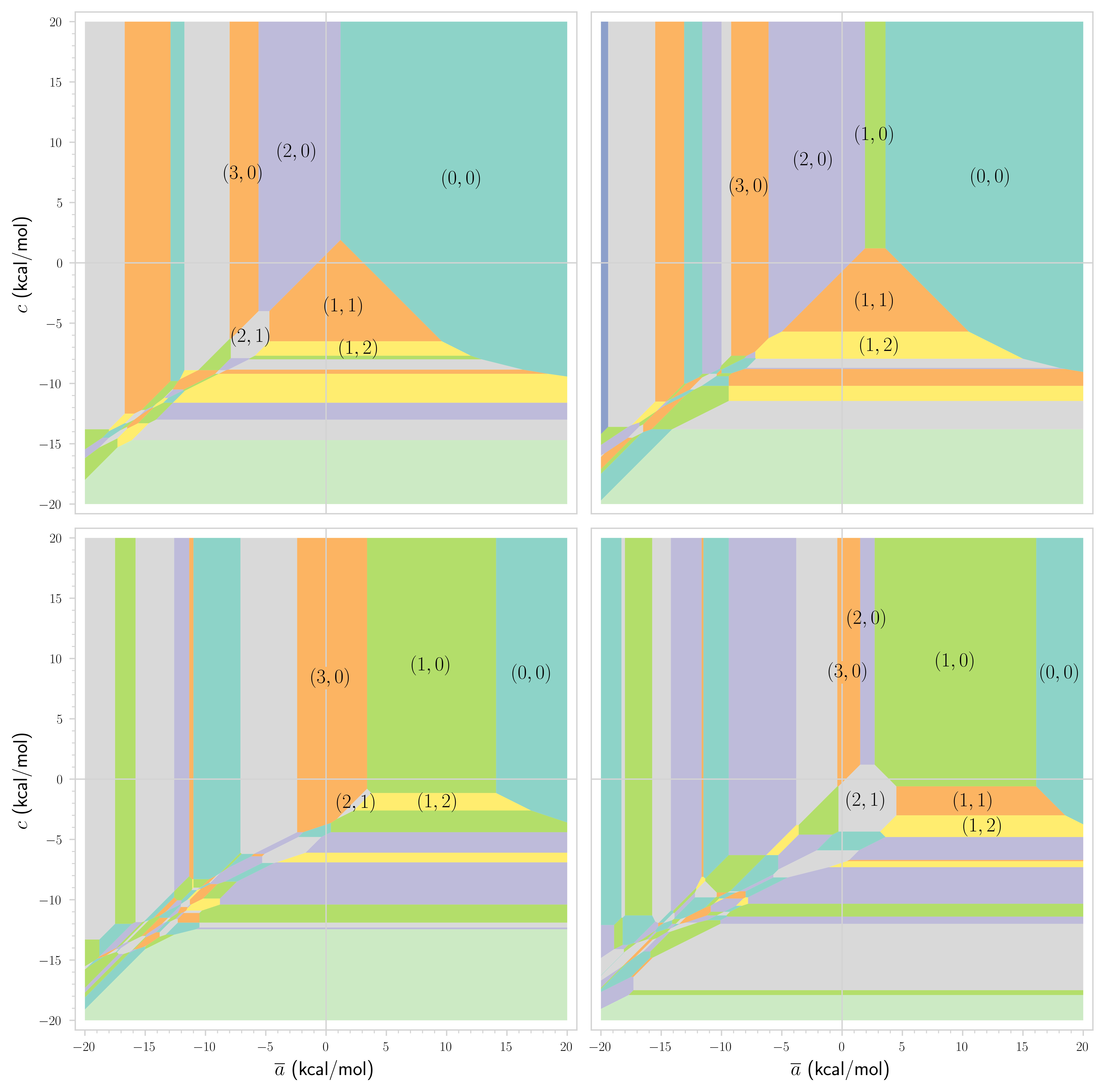}
\caption{
\edit{Four $(\na,0,c,1)$ slices with 
regions colored consistently.
The upper two are tRNA sequences, and exhibit the characteristic
large (orange) triangle/trapezoid with reduced signature $(1,1)$.  
The lower two are 5S rRNA, and have the expected large (green) 
rectangle for the $(1,0)$ region.
Within the two families,
the size and location of their target region is quite consistent.
However, the first tRNA has no $(1,0)$ region, and the second has a
more typical width.
Likewise, the first 5S rRNA has no $(1,1)$ region, and the second
has a more typical height.}
}
\label{fig:target} 
\end{figure}

We first give the geometry approximation for the $(1,0)$ region.
Assume that \reg{1}{0} exists and is bounded only by \reg{0}{0} on the right,
\reg{k}{0} for $k > 1$ on the left,
and \reg{1}{j} for $j \geq 1$ on the bottom.
Then the boundary lines are
\begin{eqnarray} 
\na & = & (\eng{0}{0} - \eng{1}{0}) \\
\na & = & \frac{(-1)(\eng{k}{0} - \eng{1}{0})}{k-1} \\
c & = & \frac{(-1)(\eng{1}{j} - \eng{1}{0})}{j} \text{.}
\end{eqnarray} 
Denote these three values as \rR, \rL, and \rB\ respectively.
Then the width of \reg{0}{0} is $\rR - \rL$.
The region is unbounded above, so we compute the height as 
$c_0 - \rB$ where $c_0 = \max_D\{\rB\} + 0.1$ 
with $D$ being the entire dataset considered.
The midpoint coordinates $(a_m, c_m)$ 
are then the averages of the two bounding line values.
When $k = 2$ and $j = 1$, this reduces to 
$\frac{1}{2}(\eng{0}{0} - \eng{2}{0}, c_0 + \eng{1}{0} - \eng{1}{1})$.

Now, we consider the geometry approximation for the $(1,1)$ region.
The assumptions for \reg{1}{1} are essentially the same.
Suppose further that $\reg{1}{0} = \emptyset$.  
Then the right, left, and bottom boundary lines are, respectively,
\begin{eqnarray} 
c & = & (-1) \na + (\eng{0}{0} - \eng{1}{1}) \\
c & = & (k-1) \na + (\eng{k}{0} - \eng{1}{1}) \\
c & = & \frac{(-1)(\eng{1}{j} - \eng{1}{1})}{(j-1)} \label{eq:11bot}
\end{eqnarray} 
and form a triangle $T$ in the $(\na,c)$-plane with apex
\begin{equation}
(\frac{\eng{0}{0} - \eng{k}{0}}{k}, 
\frac{\eng{k}{0} + (k-1)\eng{0}{0}}{k} - \eng{1}{1})\mbox{.}
\end{equation}

Let $\tL = \frac{1}{k} (\eng{k}{0} - \eng{1}{1})$, 
$\tR = \frac{k-1}{k} (\eng{0}{0} - \eng{1}{1})$,
and $\tB$ be the value from Equation~\ref{eq:11bot}.
Then the height of $T$, denoted $h_T$, is $\tL + \tR - \tB$ and
its area is $\frac{k}{2(k-1)} h_T^2$.

If instead $\reg{1}{0} \neq \emptyset$, then \reg{1}{1} is a trapezoid whose
top boundary is $c = \eng{1}{0} - \eng{1}{1}$.
Denote this value as \tT.
The trapezoid's area is $\frac{k}{2(k-1)} (h_T^2 - h_t^2)$ for 
\begin{equation}
h_t = 
\frac{1}{k}(\eng{k}{0} - \eng{1}{0}) + \frac{k-1}{k}(\eng{0}{0} - \eng{1}{0})
= \frac{k-1}{k}(\rR - \rL)\mbox{.}
\end{equation}
Hence, the area of \reg{1}{1} is a function of up to five weighted 
energy differences.

We compute the midpoint coordinates $(a_m, c_m)$ as the average of the vertices.
Let $a_0 = \frac{\eng{0}{0} - \eng{k}{0}}{k}$.
Then $a_m = a_0 + \frac{(k-2)}{3(k-1)} h_T$ if 
$\reg{1}{0} = \emptyset$
and $a_0 + \frac{(k-2)}{4(k-1)}(h_T + h_t)$ otherwise.
If $k = 2$, then this again reduces to
$\frac{1}{2}(\eng{0}{0} - \eng{k}{0})$.  
Likewise, 
$c_m = \frac{1}{3}(\tL + \tR + 2 \tB)$ in the first case
and $\frac{1}{2}(\tT + \tB)$ in the second.

Finally, we note that to generalize these equations to arbitrary $b_0 \neq 0$, 
the (weighted) differences in the $y$ terms, 
that is in the total number of unpaired bases over all junctions
between neighboring configurations,
would be added to the residual energy differences.

\subsection{\edit{Datasets}} \label{subsec:data}

\begin{table}[b!]
\centering
\begin{tabular}{l|ll|ll|ll} \hline 
\multicolumn{7}{|c|}{tRNA (min/med/max)} \\ \hline
eng & CP000493 & -48.00 & AY934351 & -32.30 & AY934184 & -23.00 \\
area & AY934254 & 0 & CP000471 & 59.34 & AP006618 & 160.70\\
f1 & AY934387\_a & 0.2439 & BX569691 & 0.7727 & AE000657 & 1\\
nntm & AE013169* & 0.0513 & CP000473* & 0.5275 & AE008623 & 1 \\
gc & X04465 & 0.4054 & DQ396875 & 0.5890 & AE014184 & 0.7297 \\
diff & BA000021* & -0.1027 & DQ093144* & 0.0357 & CP000143* & 0.7497 \\ \hline
\multicolumn{7}{|c|}{5S rRNA (min/med/max)} \\ \hline
eng & M21086 & -102.00 & X13037 & -51.60 & X06094 & -43.00 \\
area & Z75742 & 0 & K02343 & 71.33 & X07545 & 191.17 \\
f1 & V00647 & 0.1579 & M36188* & 0.7671 & X02627 & 0.9167 \\
nntm & M24954* & 0.1714 & M26976 & 0.7294 & AE009942* & 0.8941 \\
gc & U39694 & 0.4322 & X13035 & 0.5750 & X07692* & 0.7165 \\
diff & Z33604 & -0.2724 & X06996 & 0.0001 & X72588 & 0.6788 \\ \hline
\end{tabular}
\caption{Accession number and characteristic value for the bio sequences 
used as the basis for the shuf subsets.  
\edit{Six sequence characteristics were considered.
Three were properties of the target region:
the residual free energy $w$ from the branching signature (\emph{eng}), 
the total area (\emph{area}),
and the F1-accuracy of the corresponding secondary structure (\emph{f1}).
The others were the 
MFE prediction accuracy (\emph{nntm}) and GC content (\emph{gc})
originally used for sequence selection,
as well as the 
accuracy improvement over that MFE prediction for the 
target region structure (\emph{diff}).}
*~denotes the replacement for a repeated sequence. 
\label{tab:shuf}}
\end{table}

\revstart
Results are based on analyzing a number of different sets of RNA sequences.
They are broadly grouped into three categories:
\begin{enumerate}
\item the 100 biological ones used in previous studies~\cite{polystats, bnb},
\item a total of 24,000 random counterparts for those 100 sequences, and
\item the $\sim$4,000 biological sequences in the Archive II 
benchmarking dataset~\cite{sloma2016exact, mathews-19}.
\end{enumerate}

The target region geometry analysis uses the first two categories of
sequences.
The first dataset, denoted here as the \textit{bio} sequences, 
are the 50 tRNA and 50 5S rRNA whose full branching polytopes 
were previously computed and extensively analyzed~\cite{polystats, bnb}.
These sequences, and their associated secondary structures,
were originally selected from the 
Comparative RNA Web (CRW) Site~\cite{cannone-etal-02} so that
their T99 MFE prediction accuracies spanned the full range.
\revend
Differences in GC content were used to ensure that sequences with 
similar MFE values generate a diverse set.
\revstart
These 100 sequences, along with comparative secondary structures
and some associated details, can be accessed online
via 
\url{https://github.com/gtDMMB/Datasets}.
\revend
To assess biological significance,
their target region properties were compared against 
random sequences.

\revstart
The random counterparts for these bio sequences were generated 
in three different ways.
The primary comparison set, denoted as the \textit{shuf} data,
was generated using the 
\texttt{ushuffle} program~\cite{USHUFFLE}.
These shuffles are organized into 6 subsets
listed in Table~\ref{tab:shuf}.
Each subset is defined by a relevant sequence characteristic,
described in the table caption.
\revend
For each characteristic, the sequences from each family 
with the minimum, median, and maximum values were used.
In the event of a repeat, the next closest sequence was chosen.

\edit{For each of the 18 sequences chosen per family,} 
500 distinct shuffles were generated
preserving dinucleotides/2mers, i.e.\ length two substrings.
The number was based on the formula 
$n = (z \times s/e)^2$ for 
the sample size necessary to estimate the mean of a distribution,
where $s$ is the initial standard deviation, $z$ is the z-score 
corresponding to desired confidence level, and $e$ is the margin of error.  
Based on prelimary data, now included in the mers dataset described below,
it was determined that 500 samples would be sufficient 
for $z = 1.96$ (i.e.\ 95\%) and e = 0.3.

To better understand the effect of sequence length, a uniformly 
random dataset (\textit{unif}) was considered where the probabilty 
of each nucleotide is 0.25.
A set of 1500 sequences was generated with length 
74 for tRNA comparisons, and 121 for the 5S rRNA.

We also consider the effect of $k$mer length by computing a 
dataset (\textit{mers}) of 1500 sequences split evenly 
between 2mer, 3mer, and 4mer preserving shuffles of an 
\textit{Oryza nivara} tRNA sequence and also of the
5S rRNA from \textit{Escherichia coli}.

\edit{
As a consequence of the geometric analysis, we proposed 
considering alternative predictions 
generated by varying the branching parameters.
To validate the potential accuracy improvements from this approach, 
it was tested on the Archive II 
benchmarking dataset~\cite{sloma2016exact, mathews-19}.
This sequence collection was assembled by the Mathews Lab (U~Rochester)
and is available\footnote{\edit{\url{https://rna.urmc.rochester.edu/publications.html}}} online.
It consists of sequences from 10 families:
557 tRNA, 
1283 5s rRNA,
928 signal recognition particle (SRP) RNA,
454 RNase~P,
462 tmRNA,
4 domains each from 22 small subunit ribosomal RNA (16S dom),
98 group I self-splicing introns (grp~I~intr),
37 telomerase RNA,
6 domains each from 5 large subunit ribosomal RNA (23S dom),
and 11 group II self-splicing introns (grp~II~intr).
The 3 tRNA and 33 5S rRNA which happened to coincide with our training
set were removed to ensure independence of the testing dataset.
}


\section{Results} \label{sec:data}

The data analysis was evenly split between tRNA and 5S rRNA sequences.
For each family, we considered properties of 50 biological sequences 
against a large set of random comparisons.
Most were generated under permutations which preserve the 
dinucleotide/2mer distribution.
Additional randomized comparisons were made to test for dependencies 
on sequence length and choice of $k$mer.
In total, normal fan slices for 
12,000 random sequences were analyzed for each family,
grouped into
the shuf (75\%), unif (12.5\%), and mers (12.5\%)
datasets \edit{described in Section~\ref{subsec:data}}.
This was feasible using the new reduced polytope algorithm,
but still required more than 1,500 hours of compute time.

\subsection{Target region geometry}

\begin{figure}
\includegraphics[width = \textwidth]{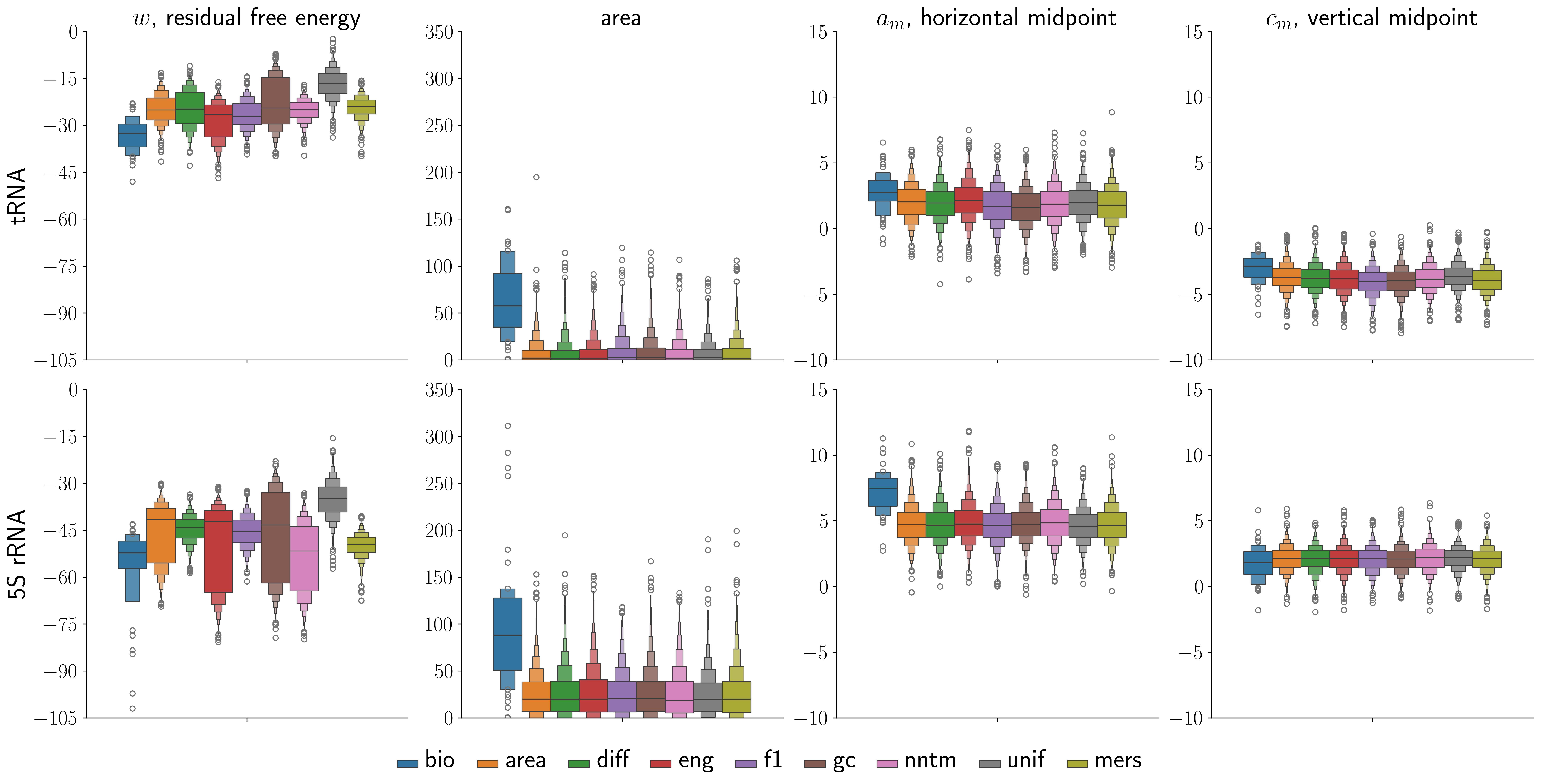}
\caption{Boxen plots~\cite{letter-value-plot} 
for the residual free energy ($w$), area, and 
center coordinates $(a_m, c_m)$ 
in the $(\na,c)$-slice, as computed by \texttt{SageMath},
for the tRNA and 5S rRNA target regions of
\reg{1}{1} and \reg{1}{0} respectively.
The central box displays the median and interquartile range, as 
in a standard boxplot.
Subsequent boxes cover a similar range to the boxplot whiskers 
but display more information.
The width of the box is exponentially correlated with the percentile
shown.
\label{fig:sage_prop}}
\end{figure}

From Figure~\ref{fig:target} \edit{in Section~\ref{subsec:geom}}, 
we see that the size and location 
of the target regions can be similar within a family but 
quite different between them.
Figure~\ref{fig:sage_prop} shows that this consistency holds 
for the randomized data, while also highlighting the very large
differences in region area between the bio sequences and their 
random counterparts.

To begin, we see different $w$ distributions
between sequences of similar lengths, 
i.e. $\sim$74 nt for tRNA and $\sim$121 nt for 5S rRNA.
This is especially true when the randomized comparisons are based
on a characteristic, like GC content, which is known to be
correlated with thermodynamic stability.

Nonetheless, these differences resolve into consistent geometric 
properties.
Hence, comparing the medians of the 8 random subsets considered,
the average (with standard deviation) for the \reg{1}{1} 
properties is 2.19 (0.40) for the area
with horizontal and vertical center coordinates of 
1.89 (0.17) and -3.84 (0.12).
The corresponding ones for \reg{1}{0} are an area of
19.99 (0.66) centered on 4.70 (0.08) horizontally and
2.15 (0.04) vertically. 
\edit{The upper bound imposed on \reg{1}{0} is $c_0 = 6.6$ since
$\max\{\rB\} = 6.5$ over the entire dataset.}

In contrast, the median bio area is 57.67 for tRNA \reg{1}{1} and 
88.115 for 5S rRNA \reg{1}{0},
and the respective center coordinates are (2.75, -2.86) and 
(7.5, 1.85).
Hence, the biological sequences have significantly different 
areas overall,
irrespective of the choice of comparison subset.
The tRNA centers are typically located  
slightly to the right and somewhat above their randomized counterparts. 
In contrast, the 5S rRNA ones are shifted 
noticably to the right, but barely down.
To understand how, and therefore why, this is happening, we 
use the geometric approximations introduced 
in \edit{Section~\ref{subsec:geom}}, conditioned on region existence.

\subsection{Existence and boundary composition}
\label{subsec:exist}

Of the 50 bio sequences for each family, 98\% have a target region 
in the $(a,0,c,1)$-plane.
When the other region also exists, they form a boundary so we 
consider these counts too.

We compare these proportions to the 9000 shuf sequences per 
family using a binomial test.
More precisely, the existence of a particular region 
is a random variable with estimated probabilty $\hat{p}$.
The p-value, denoted $p_v$, for the number of $k$ successes in $n = 50$ trials 
is reported in Table~\ref{tab:exist} under a two-tailed test.
Significance for this, and for each of the two sets of boundaries 
considered below, was assessed at $\alpha = 0.05$ under a 
Bonferroni correction.

\begin{table}
\centering
\begin{tabular}{c|*{4}{c}|*{4}{c}}
family & reg & $\hat{p}$ & $k$ & $p_v$ & reg & $\hat{p}$ & $k$ & $p_v$ \\ \hline
tRNA & 1,0 & 0.9088 & 45 & 0.8042 & 1,1 & 0.6239 & 49 & \textbf{0.0000}\\
5S rRNA & 1,0 & 0.8572 & 49 & \textbf{0.0075} & 1,1 & 0.5459 & 16 & \textbf{0.0016}\\
\hline
\end{tabular}
\caption{P-values for \reg{1}{0} and \reg{1}{1} counts
for the bio sequences compared to their shuf counterparts
under (exact) binomial test.
Bold denotes significance at $\alpha = 0.05$ under a Bonferroni correction.
\label{tab:exist}}
\end{table}

The existence of an $\reg{1}{0}$ is common, but the higher number 
for 5S rRNA is significant.
As we will see, the stability of this region affects the existence of 
an $\reg{1}{1}$, 
so there are significantly fewer than expected.  
The dual is not observed for tRNA since \reg{1}{0} is an infinite region.
Hence, although the tRNA \reg{1}{1} p-value is extreme ($2.877 \times 10^{-9}$),
the number of $\reg{1}{0}$ is typical.

We note that the estimated probabilities between the two families 
are different.
We confirmed (under an ANOVA followed by a Tukey HSD post-hoc
test against values for the unif dataset) that this is statistically 
significant and consistent with the increase in sequence length.

The right boundary of the target regions is always \reg{0}{0}, 
an infinite region which exists for every sequence.
Moreover,
there is no significant effect of the target region stability on its 
left boundary composition, 
as assessed by comparable binomial tests conditioned on the target.
All tRNA \reg{1}{1} and 79.6\% of 5S \reg{1}{0} are bounded by \reg{2}{0} 
on the left,
but this is not unusual 
($\hat{p}$, $p_v$ = 0.9421, 0.1164 and 0.8394, 0.4341).
Nor is that 89.8\% of tRNA \reg{1}{1} are bounded above by \reg{1}{0}
($\hat{p}$, $p_v$ = 0.8937, 0.8980).
However, the bottom boundaries are both significantly different.

Of the 49 tRNA \reg{1}{1}, 42 are bounded below by \reg{1}{2} and the rest
by \reg{1}{3}
($\hat{p}$, $p_v$ = 0.5777, 0.0000 and 0.3147, 0.0084).
Of the 5S rRNA, 16 are bounded below by \reg{1}{1} and 23 more by \reg{1}{2}
($\hat{p}$, $p_v$ = 0.5331, 0.0040 and 0.2666, 0.0030).
The lower boundaries for the corresponding shuffles range beyond \reg{1}{6}
(for one \reg{1}{1} and 14 \reg{1}{0}).

We note that if the number which are bounded by the first two regions
is instead considered, then the tRNA counts are still significant 
($\hat{p}$, $p_v$ = 0.8924, 0.0087) whereas the 5S rRNA are not
($\hat{p}$, $p_v$ = 0.7997, 1).
Hence the \reg{1}{1} lower boundary for biological tRNA sequences 
is  more likely to have a lower degree of branching than expected, 
whereas this is not the case for \reg{1}{0} and 5S rRNA.
However, the balance between \reg{1}{1} and \reg{1}{2} is very unusual
for 5S rRNA, most likely due to its \reg{1}{0} stability.

\subsection{Residual energies and their differences}

\begin{figure}
\includegraphics[width = \textwidth]{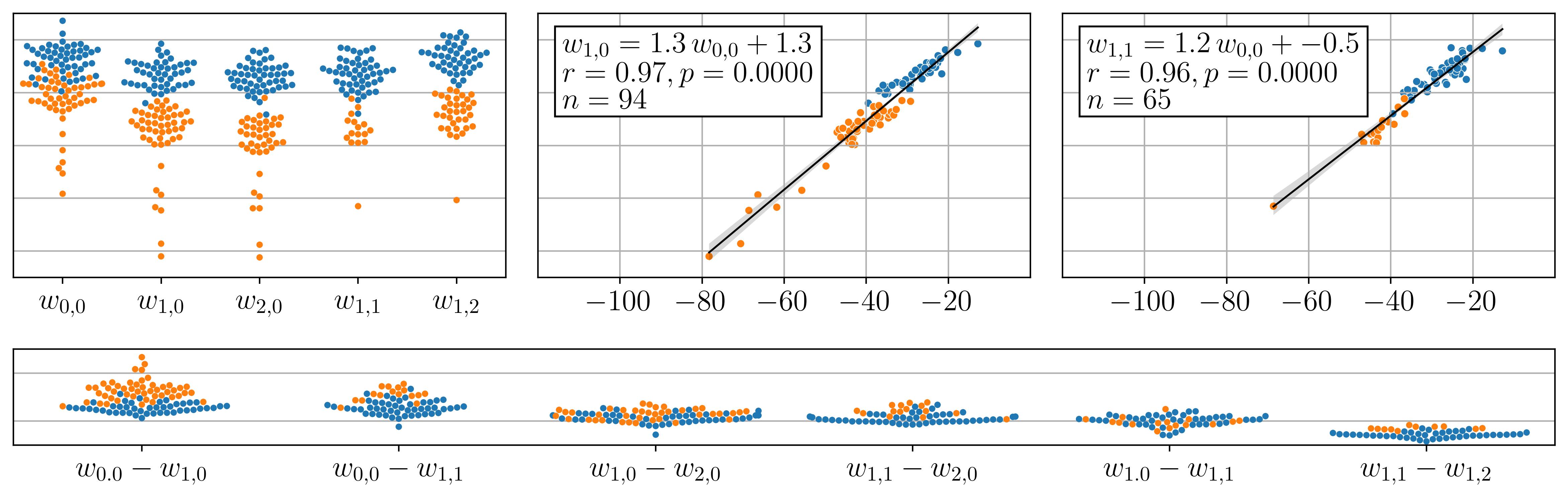}
\caption{Residual energies for the biological tRNA (blue) and 
5S rRNA (orange) sequences.
Horizontal (grey) lines are at 20 kcal/mol increments.
The best line fit was calculated with \text{scipy.stats.linregress}
with Pearson correlation coefficient $r$ and 
p-value $p$ under the Wald Test.
\label{fig:bioscat}}
\end{figure}

The geometry approximations 
derived in Section~\ref{subsec:character}
are determined by the signatures of the surrounding regions. 
The slope of the boundary lines is a function of the reduced 
signature variables, but the location (i.e.\ the axes intercepts)
depends on differences in the residual energy $w$.

As illustrated in Figure~\ref{fig:bioscat}, 
there is a range of $w$ terms for individual bio regions.
For example, the average and standard deviation for \eng{0}{0}
is (-27.38, 5.37) for tRNA and (-43.13, 10.19) for 5S rRNA yet
the difference with their target regions' residenal energy is 
(5.78, 2.94) and (12.93, 4.19) respectively.
This is due to the very strong correlation between the $w$ terms
in adjacent regions.
Figure~\ref{fig:bioscat} illustrates two such correlations for 
the most dispersed of the six differences pictured;
the other correlations are even higher.

We tested whether the strength of these correlations is 
significant by generating 10,000 random subsets of size 100 
(split evenly between the two families) from the shuf dataset 
for each of the six residual energy differences pictured.
For each subset, the Pearson correlation coefficient $r$ was
computed, generating an average and standard deviation for the 
entire dataset against which the bio $r$ values were compared
using a z-score.
Five out of the six z-scores were negative, and three 
significantly so (at $\alpha = 0.05$ under a Bonferroni correction)
whereas the one positive one, for $w_{1,0}$ versus $w_{2,0}$, was
not ($p_v = 0.2668$).
Hence, the existence of a strong correlation between neighboring 
residual energy values is not biologically meaningful.

However, the differences in those energy values can be  
significant.
Recall that the left and bottom boundary regions of the targets
are denoted \reg{k}{0} and \reg{1}{j},
and that the right boundary is always \reg{0}{0}.
As discussed in Section~\ref{subsec:exist}, 
the $k$ and $j$ values are quite consistent for the bio sequences 
but can be more variable for the shuf ones.
This is factored into the geometric approximations 
by the weighting on the residual energy differences.

Table~\ref{tab:diff} summarizes the differences between the 
bio and shuf distributions.
We see that the 5S rRNA \reg{1}{0} lower and left boundaries 
are not significantly different from random, 
but all the others are strongly so.

To further quantify this, the z-score of each bio sequence was 
computed, and summary statistics are also reported. 
An ANOVA, followed by a Tukey HSD post-hoc test, on
the tRNA z-scores and on the 5S ones confirms that the 
difference among boundaries is significant (both $p_v = 0.0000$) but
that the differences between \reg{k}{0} and \reg{1}{j} 
in both families are not ($p_v = 0.8474$ and 0.9744).
Moreover, neither are the differences between the top and right
boundaries for tRNA ($p_v = 0.3123$).

Hence, we conclude that essentially all the geometric difference 
for 5S rRNA is due to the relative stability of the (1,0) region
over the (0,0) one.
While this boundary, along with the top one, is also a major factor 
for tRNA, the left and bottom boundaries are a significant contribution
as well.

\begin{table}
\begin{tabular}{c|*{8}{c}}
family & boundary & b-avg & b-std & s-avg & s-std & p-value & z-avg & z-std \\ \hline
tRNA & \tL & -0.5765 & 0.9384 & -1.4254 & 0.8704 & \textbf{0.0000} & 0.9752 & 1.0781 \\
 & \tR & 2.7827 & 1.4532 &  0.5449 & 1.3337 & \textbf{0.0000} & 1.6778 & 1.0896 \\
 & \tB & -5.8082 & 1.1974 & -4.8590 & 1.2198 & \textbf{0.0000} & 0.7781 &  0.9817 \\ 
 & \tT & 0.2932 & 2.3240 & -2.7826 & 1.4595 & \textbf{0.0000} & 2.1074 &  1.5923\\ \hline
5S rRNA & \rL & 2.6497 & 2.1185 & 3.1197 & 1.6437 & 0.0465 & -0.2860 & 1.2889 \\
 & \rR & 12.9286 & 4.2380 & 6.6163 & 2.1807 & \textbf{0.0000} & 2.8946 & 1.9434 \\
 & \rB & -2.2081 & 2.1011 & -1.8930 & 1.4544 & 0.1319 & -0.2167 &  1.4447\\
 \hline
\end{tabular} 
\caption{
Energy difference comparison for target region boundaries.
P-values were computed under an unpaired Student t-test
on the bio (b) and shuf (s) distributions, and bolded 
when significant at $\alpha = 0.05$ under a Bonferroni correction.
The bio z-score summary statistics are also given.
\label{tab:diff}}
\end{table}

\subsection{Geometric consequences}\label{subsec:geocon}

First, 
we consider how well the geometric formula approximate the area of the 
target regions.
Recall that any difference is caused by additional boundaries in 
the lower left corner.
These decrease the actual area, and consequently shift the center coordinates.
Relatively few of the \reg{1}{0} are affected in this way.
Of the 10,314 considered, 72.66\% 	
have an absolute error in the area approximation $< 0.1$
and only 9.5\%				
are $> 1$.
Only 7 are $> 10$, and the maximum is 19.0638.
The formula are also a very good approximation for the tRNA \reg{1}{1},
although their boundaries can be more complex.
More than half (52\%) of the 7570 considered have an absolute error $< 0.1$,
with only 22.5\% being $> 1$,
53 being $> 10$ and a maximum of 26.2825.
Hence, the conclusions drawn from the residual energy comparisons
just considered are determining factors 
in the large differences in area observed.

In particular, the size of the \reg{1}{0} target region for 5S rRNA 
is driven by its stabilty relative to \reg{0}{0}.
This results in a much larger width, but no  meaningful
difference in height.
As a consequence, the center is shifted horizontally, 
but not vertically.
Yet, the existence of an \reg{1}{1} was found to be  significantly 
different for the bio sequences.

To gain some insight,
the \rB\ values were considered only for those bio and shuf sequences
with $j > 1$.
These were compared with the 5S rRNA \tB\ values for bio and shuf 
datasets.
An ANOVA followed by a Tukey HSD post-hoc test found that  
the two shuf distributions are significantly different
($p_v = 0.0000$) from each other 
but the two bio ones are not ($p_v = 0.9978$).
Hence, for the bio sequences but not the shuf ones, 
the \reg{1}{0} bottom when $j > 1$ is very close to where 
the \reg{1}{1} bottom is likely to have been.
Moreover, the bio \rB\ value when $j = 1$ is significantly 
higher (-0.4313 versus -3.0695, $p_v = 0.0000$) than not.
Thus, more 5S rRNA do not have a \reg{1}{1} than expected.

In contrast, all four boundaries contribute to the size 
of the tRNA \reg{1}{1}. 
As a consequence the effect on the center location is less pronounced,
although it is shifted right and up, consistent with 
the greater contributions of \tR\ and \tT\ than \tL\ and \tB.
Due to the slope of the boundary lines, the width is closely 
related to the ``full'' triangle height $h_T = \tR + \tL - \tB$.
We find, though, that the distribution of $h_t$ is quite similar
to the random ones. 
In other words, the stability of \reg{1}{1} relative to \reg{1}{0}
results in essentially the same amount being deducted from 
the full triangle area as for the random sequences when $\reg{1}{0}$
exists.
This is consistent with the tRNA \reg{1}{0}
being like their random counterparts.

\subsection{Prediction} \label{subsec:prediction}

\revstart

Previous work~\cite{polystats} demonstrated that understanding the polytope geometry can help improve predictions under NNTM optimization. The ad-hoc parameters determined by analyzing intersections of regions of interest for our 100 tRNA and 5S rRNA sequences showed improved accuracy across both families. However, since we also considered families separately, we observed that the combined accuracy was well below what could be achieved with family-specific parameters.

Subsequent development of a branch-and-bound algorithm~\cite{bnb} to assess  potential mprovements across families confirmed that the ad-hoc parameters performed as well as it could be expected. Both ad-hoc and branch-and-bound parameters led to improved predictions on a larger testing set of tRNA and 5S rRNA sequences from the Archive II dataset~\cite{sloma2016exact, mathews-19}, with ad-hoc parameters outperforming, likely due to reduced overfitting to the training data.

Here, we adopt the same approach used in~\cite{polystats} to identify the ad-hoc parameters. We construct the graph of pairwise intersections of target regions across the entire training set of 100 sequences, as well as each family separately. In this graph, two regions are connected if they intersect. We then identify maximal cliques and take cumulative intersections over those cliques. The aim is to find a single parameter combination, but we consider each family separately in order to develop understanding how the accuracy of each is affected by inclusion of the other in the training set. As suggested by the example in Figure~\ref{fig:target}, pairwise intersections across families are easily found. The problem is that their locations vary, making finding cumulative intersections of multiple regions more challenging. 

Each family was found to have only a few maximal cliques, each with a nonempty intersection. Moreover, for each family we identified a large ``core'' which belongs to all maximal cliques containing over 80\% of the sequences. As illustrated in Figure~\ref{fig:regions_plot}, the core intersections of tRNA and 5S rRNA are distinctly separated, even exhibiting opposite preferences for the sign of the $c$-parameter.

\begin{figure}[]
\centering     \includegraphics[width=\linewidth]{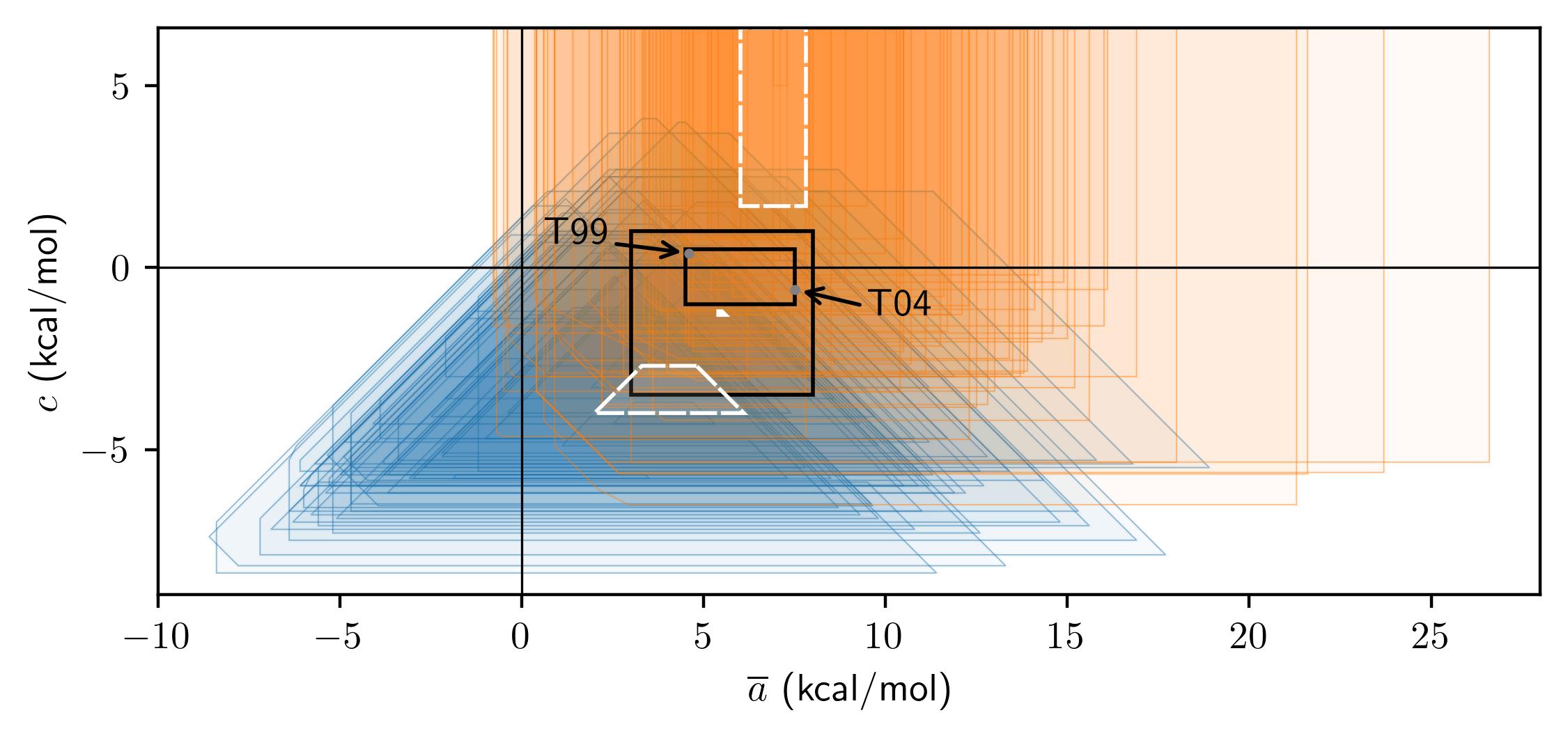}
\caption{\edit{The target regions for tRNA (blue trapezoids) are overlayed with those for 5S rRNA (orange rectangles). Regions with white boundaries represent the cumulative intersections of the target regions for a large core identified within each family, as well as the most balanced clique across families -- 42  tRNA (bottom), 44 5S rRNA (top) and  31 tRNA and 35 5S rRNA (middle). The search rectangles Trec (smaller) and Grec (larger) are shown in black, with the unskewed versions of the T99 and T04 parameters marked inside. The precise region locations are as follows. tRNA : vertices $(4.8, -2.7)$, $(3.3, -2.7)$, $(2.0, -4.0)$, $(6.1, -4.0)$, 5S rRNA: $6 \leq \na \leq 7.8$, $c \geq 1.7$, tRNA/5S rRNA: vertices $(5.4, -1.3)$,$(5.4, -1.2)$,$(5.5, -1.2)$,$(5.6, -1.3)$, Trec: $4.5 \leq \na \leq 7.5$, $-1 \leq c \leq 0.5$, Grec: $3 \leq \na \leq 8$,  $-3.5 \leq c \leq 1$.}}
\label{fig:regions_plot}
\end{figure}

In contrast, the entire set’s intersection graph contained numerous maximal cliques of varying sizes, with most cumulative intersections being empty. The largest subsets with nonempty intersections include only about two-thirds of the sequences and some of them are skewed toward tRNA.  The cumulative intersection for the most balanced clique between families (Figure~\ref{fig:regions_plot}) is small and disjoint from both core intersections.  Therefore, while each family’s prediction can be improved with different parameter sets, using a single parameter set results in lower accuracy than family-specific parameters.

Given the evidence that the best linear approximation of the branching term varies between families, we suggest considering optimal structures under multiple linear approximations as a basis for predictions, rather than choosing a single parameter set. We test this idea on the Archive II dataset, sampling parameters from two rectangles in 0.5 increments. The first rectangle we consider (Trec) is the smallest rectangle containing the T99 and T04 (unskewed) parameters, while the second rectangle (Grec) was selected based on the geometry of the 100 target regions in our training set. Bounds for the $\na$ values were chosen based on the ranges of the $\na$ values of the centers of the target regions in both families and is not much bigger than the Trec $\na$ range. Since the 5S RNA regions are infinite vertical strips, with the tRNA trapezoidal regions being generally lower than them, the $c$ range was chosen based on the the top decile of the bottom boundaries for 5S RNA and the bottom decile of the top boundaries for tRNA,  rounded to the nearest 0.5 increment and expanded to include an extra row.

Table ~\ref{tab:predictionrecomputed} shows significant improvements in prediction accuracy when, instead of using only T04, multiple structures optimal under branching parameters from Trec are considered. We report accuracy for each family separately because the test data is highly unbalanced. We observe that the number of structures generated generally remains within single digits for the sequence lengths in our test set. Further improvements were observed when sampling from Grec (significant under a Bonferroni correction) for all families except the group II intron, likely due to sample size. Improvements for tRNA and 5S rRNA sequences from the Archive II data indicate that we are not overfitting to our training sequences when considering the larger rectangle. 

\begin{table}[]
\revstart
\centering
\begin{tabular}{c|cc|c|ccc|ccc}
\multicolumn{3}{c|}{sequences} & T04 & \multicolumn{3}{c|}{Trec} & \multicolumn{3}{c}{Grec} \\ \hline
family & \makecell{seq \\ len} & \makecell{set \\ size} & acc  & acc & \makecell{p-value\\ vs T04} & \makecell{structures \\ (avg $\pm$ std)}  &  acc & \makecell{p-value  \\ vs Trec} & \makecell{structures \\ (avg $\pm$ std)}  \\ \hline
tRNA* & 77 & 555 & 0.53 & 0.68 & \textbf{0.0000} & $2 \pm 1$ & 0.79 & \textbf{0.0000} & $3 \pm 1$   \\
5S rRNA* & 119 & 1250 & 0.62 & 0.66 & \textbf{0.0000} & $2 \pm 1$ & 0.68 & \textbf{0.0000} & $4 \pm 2$
\\ 
SRP  & 184 & 928 & 0.67 & 0.69 & \textbf{0.0000} & $2 \pm 1$ & 0.70 & \textbf{0.0000} & $5 \pm 4$  \\ 
RNaseP & 332 & 454 & 0.52 & 0.61 & \textbf{0.0000} & $5 \pm 2$ & 0.66 & \textbf{0.0000} & $17 \pm 5$   \\ 
tmRNA & 366 & 462 & 0.40 & 0.49 & \textbf{0.0000} & $6 \pm 2$ & 0.53 & \textbf{0.0000} & $17 \pm 5$ \\ 
16S dom & 378 & 88 & 0.53 & 0.59 & \textbf{0.0000} & $4 \pm 3$ & 0.61 & \textbf{0.0001} & $15 \pm 9$ \\ 
grp I intr & 426 & 98 & 0.51 & 0.57 & \textbf{0.0000} & $6 \pm 3$ & 0.61 & \textbf{0.0000} & $20 \pm 4$\\ 
telomerase & 445 & 37 & 0.49 & 0.55 & \textbf{0.0000} & $5 \pm 2$ & 0.58 & \textbf{0.0013} & $16 \pm 4$ \\ 
23S dom & 461 & 30 & 0.66 & 0.73 & \textbf{0.0008} & $6 \pm 3$ & 0.76 & \textbf{0.0032} & $19 \pm 9$ \\
grp II intr & 717 & 11 & 0.28 & 0.34 & \textbf{0.0003} & $11 \pm 3$ & 0.37 & 0.0417 & $37 \pm 8$  \\
\hline
\end{tabular} 
\caption{\edit{Per family prediction accuracy on the Archive II data set (*the tRNA and 5S rRNA sequences which happened to coincide with the training set were excluded from the test set). The first three columns describe the composition of the data set. Comparisons are given between the average accuracy (\emph{acc}) of the optimal structure under the T04 parameters (measured as the F1-score) and the most accurate structure found within the Trec rectangle, as well as the best structures within the extended Grec rectangle. The p-value from a paired t-test is given for each family and bolded when significant at $\alpha = 0.05$ under a Bonferroni correction. The number of different structures seems to grow linearly with sequence length.}}
\label{tab:predictionrecomputed}
\revend
\end{table}

Based on the existence of a much better prediction in a not much bigger set of structures, this is a promising new way for sampling structures for prediction. We discuss possible directions in which this method can be developed to generate a prediction in the Conclusions section. 

\revend


\section{Discussion} \label{sec:discuss}

\revstart

Previous analysis of branching parameters which improve the MFE prediction indicated that different parameter combinations favor different families. Here we focused on the two families which were analyzed previously -- tRNA and 5S rRNA -- and on the regions made of the parameters which yield the same optimal structure for each RNA sequence, rather than parameter combinations. 

When compared to their dinucleotide shuffles, the existence of the target branching region is observed more frequently for the biological sequences. We find statistically that the region existence is partly affected by sequence length and it would be interesting to find a mathematical explanation for this fact. However, we find that this does not explain the differences for the biological sequences -- they have the target region more frequently than expected for their length. 

To analyze the geometry of the target regions for these families, it is useful to perform a change of variables and consider only the excess instead of the total amount of branching in the secondary structure. While this formulation cannot be used directly in the optimization since it is not amenable to  subdivision, it eliminates the 1/3 skew in the original parameter space and makes the regions more symmetric.  The transformation is invertible, so any conclusion obtained about the parameters in the unskewed space translates directly to an analogous one for the original parametrization. 

We give a geometric description of the target branching regions for tRNA and 5S rRNA and we show that they are well approximated by a triangle or trapezoid (for tRNA) and an infinite  vertical strip (for 5S rRNA). The formulas we derive imply that the region location and boundaries depend on the relative difference between the residual energies of the target branching structure and only 3 or 4 other branching patterns. 

While the geometry of the target region is intrinsic to the NNTM optimization, its size is due to the biological signal. It is known that functional RNAs have lower folding energy than random RNAs of the same length and dinucleotide frequency~\cite{clote2005structural}. However, the residual energies for different regions are correlated, so this fact does not imply difference in the regions of different biological and random sequences. Interestingly, the statistical analysis of the residual differences (due to the stackings and other loops) of the target and neighboring structures for tRNA and 5S rRNA showed that the distributions are similar for each family and are a new type of biological signal. This explains why the target regions within each family have a large overlap and why we can improve the prediction for each family separately.

We also showed that while the overlap of the target regions across families is smaller than within each family, the location we observe for tRNA and 5S rRNA indicates an area which might be useful for other families of sequences, if one is open to considering other possible optimal structures by varying the branching parameters.  With the rectangle search, the prediction improved for all of the families in the testing set, while the set of structures generated wasn't very large and appeared to grow linearly with sequence length. It might be of interest to be able to consider only the proposed search rectangle here or to further ``zoom into'' a related area.

\revend


\section{Conclusions} \label{sec:conclusion}

\revstart

This work answers open questions while raising new ones. We now understand why multiple parameter optimization approaches have concluded that it is impossible to  improve predictions for both tRNA and 5S rRNA simultaneously~\cite{ward2017advanced, polystats}. Significant improvements for both families are achievable, but only for different parameters, and as previous work has shown~\cite{bnb}, this disparity is not due to either family being an outlier. Some of the other families in the Archive II set show improved predictions with tRNA-like parameters, while others benefit more from 5S rRNA-like parameters.

As we observed, substantial gains in accuracy can be achieved by considering multiple branching configurations. This approach suggests a new direction to explore for improving NNTM optimization. It also raises a key question: how should we handle the structures generated? Multiple approaches such as summarizing the data~\cite{ding2001statistical, rogers2014profiling}, re-evaluating under more sophisticated energy models~\cite{liu2011fluorescence, zuber2024estimating}, or incorporating chemical footprinting data~\cite{ge2015computational} have been successfully applied in other prediction methods that rely on large samples,  making them natural next steps to explore. 

To generate the data, we developed an algorithm for generating a fixed $b = b_0, d=1$ slice of the normal fan of the branching polytope for each sequence. The algorithm is faster than considering the whole branching polytope, as it effectively disregards a large proportion of secondary structures which do not correspond to the parameters in the slice of interest. While effective for our calculations, its complexity remains prohibitive for practical use for longer RNA sequences. Therefore, it would be advantageous to develop an algorithm that generates only a part of the slice, which would further speed up the computations by excluding a significant number of irrelevant structures for longer sequences. We plan to pursue this in future work.

\revend

\end{document}